\documentclass[11pt]{article}

\usepackage{cite}
\usepackage{graphicx}
\usepackage{amsmath}
\usepackage{amsthm}
\usepackage{amsfonts}
\usepackage{amssymb}
\usepackage{fullpage}
\usepackage{geometry}
\usepackage{color}
\usepackage{latexsym}
\usepackage{algorithm}
\usepackage[noend]{algorithmic}
\usepackage{setspace}
\usepackage{subfigure}
\usepackage{soul}
\usepackage{authblk}

\newtheorem{theorem}{Theorem}[section]
\newtheorem{corollary}[theorem]{Corollary}

\newtheorem{myclaim}[theorem]{Claim}

\def\marrow{\marginpar[\hfill$\longrightarrow$]{$\longleftarrow$}}

\def\rom#1{\textsf{(Rom says: \marrow{#1})}}

\def\gila#1{\textsf{(Gila says: \marrow{#1})}}

\newcommand{\old}[1]{{{}}}

\def\segment#1#2{{\overline{#1#2}}}

\def\UDG{{U\!DG}}
\def\VR#1{{{V\!R}_{#1}}}
\def\VP#1{{{V\!P}_{#1}}}

\newcommand{\reals}{\ensuremath{\mathbb{R}}}

\begin{document}

\title{Approximation Schemes for Covering and Packing\thanks{%
Work by R. Aschner was partially supported by the Lynn and William Frankel Center for Computer Sciences. Work by R. Aschner and M.J. Katz was partially supported by the Israel Ministry of Industry, Trade and Labor (consortium CORNET).
Work by M.J. Katz was partially supported by grant 1045/10 from the Israel Science Foundation, and by grant 2010074 from the United States -- Israel Binational Science Foundation.
Work by G. Morgenstern was partially supported by the Caesarea Rothschild Institute (CRI).}}



\author[1]{Rom Aschner}
\author[1]{Matthew J. Katz}
\author[2]{Gila Morgenstern}
\author[1]{Yelena Yuditsky}

\affil[1]{{\small Department of Computer Science, Ben-Gurion University}}
\affil[ ]{{\small {\tt \{romas,matya,yuditsky\}@cs.bgu.ac.il}}}

\affil[2]{{\small Caesarea Rothschild Institute, University of Haifa} 
}
\affil[ ]{{\small {\tt gilamor@cri.haifa.ac.il}}}

\maketitle

\begin{abstract}
The local search framework for obtaining PTASs for NP-hard geometric optimization problems was introduced, independently, by Chan and Har-Peled~\cite{ChanH09} and Mustafa and Ray~\cite{MR10}. In this paper, we generalize the framework by extending its analysis to additional families of graphs, beyond the family of planar graphs. We then present several applications of the generalized framework, some of which are very different from those presented to date (using the original framework).
These applications include PTASs for finding a \emph{maximum $l$-shallow set} of a set of fat objects, for finding a \emph{maximum triangle matching} in an $l$-shallow unit disk graph, and for \emph{vertex-guarding} a (not-necessarily-simple) polygon under an appropriate shallowness assumption.

We also present a PTAS (using the original framework) for the important problem where one has to find a minimum-cardinality subset of a given set of disks (of varying radii) that covers a given set of points, and apply it to a class cover problem (studied in~\cite{BeregCDPSV10}) to obtain an improved solution.
\end{abstract}

\section{Introduction}
In their break-through papers, Chan and Har-Peled~\cite{ChanH09} and Mustafa and Ray \cite{MR10}
showed, independently, how a simple local-search-based algorithm can be employed to obtain a PTAS for an NP-hard geometric optimization problem.
Chan and Har-Peled \cite{ChanH09} used local search
to obtain a PTAS for finding a maximum independent set of pseudo disks, and
Mustafa and Ray \cite{MR10} used it to obtain a PTAS for finding a minimum
hitting set (from a given set of points) for half-spaces in $\reals ^3$ and for $r$-admissible regions in $\reals ^2$.
This technique turned out to be very powerful, and since its publication, it was applied to a variety of additional problems.
Gibson  et al.~ \cite{GKKV09} used it to obtain a PTAS for the $1.5$D terrain guarding problem,
and Gibson and Pirwani used it to obtain a PTAS for finding a dominating set in disk graphs~\cite{GP10}.

The local-search-based algorithm, as described in~\cite{ChanH09,MR10}, receives an integer parameter $k$ and proceeds as follows.
It starts with any feasible solution and performs a series of local improvements, where each such improvement involves only $O(k)$ objects.
The analysis relies on the existence of a
planar bipartite graph $G$, whose vertices on one side correspond to the objects found by the local search algorithm (``blue vertices'')
and on the other side to the objects in an optimal solution (``red vertices''), and whose (``blue-red'') edges
satisfy an appropriate \emph{locality property} relating the two solutions.
Chan \& Har-Peled and Mustafa \& Ray showed that for the problems mentioned above such a planar graph $G$ exists,
and applied the planar separator theorem to
relate the size of the local search solution with that of the optimal solution.

In this paper we generalize the local search technique by extending its analysis to additional families of graphs.
We show that one can achieve a PTAS even in cases where the graph whose vertices are the elements of the two solutions
and which satisfies the locality property, is \textbf{not} planar, but belongs to a family of graphs that has a separator property.
It is well known that there are many
such families of graphs,
e.g., graphs with forbidden minors~\cite{ASTforbiddenminor} and
various intersection graphs~\cite{FPseparator,MillerTTV97,SmithW98}.
We also present several interesting applications of our extended analysis, some of which are
very different from those presented in the past (using the original analysis).

These applications, which are presented in Section~\ref{sect:apps}, include finding a \emph{maximum $l$-shallow set} of a set of fat objects,
finding a \emph{maximum triangle matching} in a unit disk graph, and \emph{guarding a polygon with limited range of sight}. We briefly discuss each of them.

\noindent \textbf{\textit{Maximum $l$-shallow set}}.
Let $D$ be a set of $n$ fat objects in $\reals ^2$ and let $l > 0$ be a constant. An {\em $l$-shallow}
subset of $D$ is a subset $S$ of $D$, such that the depth of the arrangement of $S$ is at most $l$, i.e. every point in the plane is covered by at most $l$ objects of $S$.
In the maximum $l$-shallow subset problem, one is asked to find a maximum-cardinality $l$-shallow subset $S$ of $D$.
Notice that for $l=1$, $S$ is a maximum independent subset of $D$.
Chan~\cite{Chan03} presented a PTAS for the problem of finding a maximum independent subset of $D$.
Later, as mentioned above, Chan and Har-Peled~\cite{ChanH09} presented a local-search-based PTAS for finding
a maximum independent set of pseudo-disks (and also of $D$).
Notice that a maximum $l$-shallow set can be larger than the set that is obtained by repeatedly finding a maximum independent set of the set of remaining objects.
We show that our generalized analysis enables us to obtain a PTAS for the maximum $l$-shallow subset problem, and emphasize that it is essential whenever $l > 1$.

\noindent \textbf{\textit{Maximum triangle matching}}.
Given a graph $G$, one has to find a maximum-cardinality collection of vertex-disjoint triangles in $G$.
Baker~\cite{B94} presented a PTAS for the important case where $G$ is planar.
We give a PTAS for the case where $G$ is an $l$-shallow unit disk graph. This PTAS also holds in more general settings, see below.

\noindent \textbf{\textit{Guarding}}.
We apply our generalized analysis to several guarding problems.
Many guarding problems are known to be APX-hard (e.g., guarding a simple $n$-gon $P$ with as few vertex guards as possible~\cite{ESW01}), and as such do not admit a PTAS.
Ghosh~\cite{Go87} presented an $O(\log n)$-approximation algorithm for vertex-guarding a polygon. Subsequently, Efrat and Har-Peled~\cite{EH-P06} presented an $O(\log |OPT|)$-approximation algorithm for this problem. In contrast to these results, we obtain a PTAS for this problem, under the assumption that each vertex guard $g$ has a limited range of sight $r_g$, such that every point in $P$ is covered by at least one of the guards and the set of disks centered at the guards is $l$-shallow.
We also discuss several other versions, including guarding through walls and guarding a 1.5D terrain with bounded range of sight.

Finally, in Section~\ref{sec:coverage}, we consider the important problem known as discrete coverage of points by disks, for which we obtain a PTAS within the original local search framework.
Let $P$ be a set of points in the plane and let $D$ be a set of disks that covers $P$.
One needs to find a minimum-cardinality subset $D' \subseteq D$ that covers $P$.
Notice that the special case where $D$ is a set of unit disks, is the dual of the discrete hitting set problem for unit disks,
for which Mustafa and Ray~\cite{MR10} presented a PTAS (actually, for arbitrary disks).
However, in the general case, where $D$ is a set of disks of varying radii,
coverage and hitting are not dual. 
We present a PTAS for the general case that is inspired by the work of Gibson and Pirwani~\cite{GP10}.
We also apply this result to the a class cover problem, improving a result of Bereg et al.~\cite{BeregCDPSV10}.

\section{PTAS via local search}\label{sec:general}
We begin by generalizing the local search technique to additional families of graphs, beyond the family of planar graphs, such as intersection graphs and graphs with forbidden minors. Actually, we describe the technique for any family of graphs that has a separator property, similar to the separator property of planar graphs.

Let $\cal {F}$ be a monotone family of graphs, i.e., all subgraphs of a graph $G \in {\cal F}$ are also in ${\cal F}$. Assume that ${\cal F}$ has a separator property, i.e., for any graph $G=(V,E)$ in ${\cal F}$, one can partition the vertex set $V$, where $|V|=n$, into three sets $A$,$B$ and $S$, such that (i) $|S| \le c n^{1-\delta}$, where $0 < \delta < 1$ and $c > 0$ is a constant, (ii) $|A|,|B| \le \alpha n$, where $1/2 \le \alpha < 1$, and (iii) the sets $A$ and $B$ are disconnected, i.e., there is no edge in $E$ between a vertex of $A$ and a vertex of $B$.
\old{
\begin{itemize}
\item $|S| \le c n^{1-\delta}$, where $0 < \delta < 1$ and $c > 0$ is a constant.
\item $|A|,|B| \le \alpha n$, where $1/2 \le \alpha < 1$.
\item The sets $A$ and $B$ are disconnected, i.e., there is no edge in $E$ between a vertex of $A$ and a vertex of $B$.
\end{itemize}
}

Frederickson~\cite{Fred87} defined the notion of an $r$-division for planar graphs, and showed how to obtain an $r$-division for a given planar graph $G$, by repeatedly applying the planar separator algorithm. It is straight-forward to adapt Frederickson's construction and analysis to the family ${\cal F}$. Essentially, one needs only to replace the exponent $1/2$ in the size of the separator by $1-\delta$.
Thus, ${\cal F}$ has the following property. Let $r$ be a parameter, $1 \le r \le n$, then for any connected graph $G=(V,E)$ in ${\cal F}$, one can find a collection of $\Theta(n/r)$ pairwise disjoint subsets $V_1, V_2, \ldots$ of $V$, such that
(i) $|V_i| \le c_2 r$, where $c_2 > 0$ is a constant, (ii) $|\Gamma(V_i)| \le c_3 r^{1-\delta}$, where $\Gamma(V_i)$ is the set of vertices in $V \setminus V_i$ that are adjacent to a vertex in $V_i$ and $c_3 > 0$ is a constant, and (iii) $\cup \Gamma(V_i) = S$, where $S = V \setminus \cup V_i$;  in particular, $\Gamma(V_i) \subseteq S$, and, for any $j \ne i$, $V_i$ and $V_j$ are disconnected.
\old{
\begin{itemize}
\item $|V_i| \le c_2 r$, where $c_2 > 0$ is a constant.
\item $|\Gamma(V_i)| \le c_3 r^{1-\delta}$, where $\Gamma(V_i)$ is the set of vertices in $V \setminus V_i$ that are adjacent to one or more vertex in $V_i$ and $c_3 > 0$ is a constant.
\item $\cup \Gamma(V_i) = S$, where $S = V \setminus \cup V_i$;  in particular, $\Gamma(V_i) \subseteq S$, and, for any $j \ne i$, $V_i$ and $V_j$ are disconnected.
\end{itemize}
}
It follows that $|S| \le \Theta(n/r) \cdot c_3 r^{1-\delta} = c_1 \frac{n}{r^\delta}$, where $c_1 > 0$ is a constant.
Also, it is easy to verify that
\begin{myclaim}
\label{claim:size}
If $r=\frac{k}{c_2+c_3}$, where $k$ is a sufficiently large constant, then,
for any index $i$, $|V_i| + |\Gamma (V_i)| < k$.
\end{myclaim}
\old{
\begin{proof}
\old{
\begin{eqnarray}
\lefteqn{|V_i| + |\Gamma(V_i) |\leq c_2 r + c_3 r^{1-\delta} = \frac{c_2 k}{c_2 +c_3} + \frac{c_3 k^{1-\delta}} {(c_2+c_3)^{1-\delta}}} \nonumber \\
&=& \frac{c_2 k + c_3  (c_2 +c_3)^\delta  k^{1-\delta}}{c_2+c_3} < \frac {c_2 k + c_3 k}{c_2 + c_3}  = k \ . \nonumber
\end{eqnarray}
}

\begin{align*}
|V_i| + |\Gamma(V_i) | &\leq c_2 r + c_3 r^{1-\delta} = \frac{c_2 k}{c_2 +c_3} + \frac{c_3 k^{1-\delta}} {(c_2+c_3)^{1-\delta}} \\
&= \frac{c_2 k + c_3  (c_2 +c_3)^\delta  k^{1-\delta}}{c_2+c_3} < \frac {c_2 k + c_3 k}{c_2 + c_3}  = k \ .
\end{align*}

\end{proof}
}
We distinguish between minimization problems and maximization problems.

\subsection{Minimization problems}
\label{sect:minimization_problems}
Let us now recall the local search technique as it is used, e.g., in the context of geometric piercing or covering.
Assume that we are considering a minimization problem ${\cal P}$, that is, one needs to find a minimum-cardinality subset of a given set $X$ that is a solution for ${\cal P}$.
Let $X_0 \subseteq X$ be an initial not-necessarily-optimal solution for ${\cal P}$, and let $k$ be a sufficiently large constant.
The local search technique checks whether there exists a subset $X' \subseteq X_0$ of size $k$ and a subset $X'' \subset X$ of size $k-1$, such that $(X_0 \setminus X') \cup X''$ is still a solution for ${\cal P}$. If yes, then it replaces $X'$ by $X''$ (i.e., it performs this local improvement) and resumes the search. Otherwise, it halts.

Let $B$ and $R$ be two solutions for ${\cal P}$, where $B$ was obtained by applying local search, and assume that $B \cap R = \emptyset$. (Otherwise, we can remove the elements that belong to both $B$ and $R$.)
Moreover, let $G=(V,E)$ be a graph, such that (i) $V = B \cup R$, and (ii) for each object $o$ that needs to be ``solved'', there exists an edge $e \in E$ between a vertex $b \in B$ and a vertex $r \in R$, where both $b$ and $r$ ``solve'' $o$. This requirement is called the \emph{locality condition}.
We prove below that if $G$ belongs to some family ${\cal F}$ that has a separator property, then $|B|$ is not much greater than $|R|$.
Parts of the proof appear already in~\cite{ChanH09,MR10} and are included here for completeness.
We first construct an $r$-division of $G$, for $r = \frac{k}{c_2+c_3}$, and set $B_i= B \cap V_i$ and $R_i= R \cap V_i$.

\begin{myclaim}
\label{claim:substition}
For any index $i$, the set $(B \setminus B_i) \cup \Gamma(B_i)$ is also a solution for ${\cal P}$.
\end{myclaim}
\begin{proof}
Fix $i$ and let $o$ be an object that needs to be ``solved''. If all vertices of $B$ that ``solve'' $o$ belong to $B_i$, then, by the locality condition, there exists $r \in \Gamma(B_i)$ that ``solves'' $o$. Otherwise, there is a vertex $b \in B \setminus B_i$ that ``solves'' $o$. We conclude that in both cases there is a vertex in $(B\setminus B_i) \cup \Gamma (B_i)$ that ``solves'' $o$.
\end{proof}

\begin{myclaim}
For any index $i$, $|B_i| \le |R_i|+|\Gamma(V_i)| < k$.
\end{myclaim}
\begin{proof}
Observe first that by arguments similar to those in the proof of Claim~\ref{claim:substition}, the set $(B \setminus B_i) \cup (R_i \cup \Gamma(V_i))$ is also a solution for ${\cal P}$. Moreover, by Claim~\ref{claim:size}, $|R_i|+|\Gamma(V_i)| \leq |V_i|+ |\Gamma (V_i)| < k$. So, if $|B_i|>|R_i|+|\Gamma(V_i)|$, then the local search algorithm would have replaced $B_i$ (or a subset of $B_i$ of size $k$, if $|B_i| > k$) by $R_i \cup \Gamma(V_i)$ before halting. Since it has not done so, we conclude that $|B_i| \le |R_i|+|\Gamma(V_i)|$.
\end{proof}

\begin{theorem}\label{thm:ptas}
For any $\varepsilon$, $0 < \varepsilon < 1$, one can choose a constant $k$, such that
$|B|\le(1+\varepsilon)|R|$.
\end{theorem}
\begin{proof}
Set $c'=2c_1(c_2+c_3)^\delta$, $\varepsilon' = \frac {\varepsilon}{5}$, $k=\frac {c'^{1 / \delta}}{{\varepsilon ' } ^ {1/ \delta} }$, and recall that $r=\frac {k}{c_2+c_3}$. Then,
\old{
\begin{eqnarray}
\lefteqn{|B|\leq |S|+\sum_i |B_i| \leq |S| + \sum_i |R_i| + \sum _i |\Gamma(V_i)|} \nonumber \\
& \leq& \frac {c_1 n}{r^\delta} + |R|+ \frac {c_1 n}{r^\delta} = |R|+ \frac {2c_1 n}{r^\delta} \nonumber \\
&=& |R|+ \frac {|R|+|B|}{k^ \delta} 2c_1(c_2+c_3)^\delta  \nonumber \\
&=& |R|+\varepsilon' (|R|+|B|) \ . \nonumber
\end{eqnarray}
}
\old{
\begin{align*}
|B| &\leq |S|+\sum_i |B_i| \leq |S| + \sum_i |R_i| + \sum _i |\Gamma(V_i)| \leq \frac {c_1 n}{r^\delta} + |R|+ \frac {c_1 n}{r^\delta} \nonumber \\
&= |R|+ \frac {2c_1 n}{r^\delta} = |R|+ \frac {|R|+|B|}{k^ \delta} 2c_1(c_2+c_3)^\delta = |R|+\varepsilon' (|R|+|B|) \ . \nonumber
\end{align*}
}
$|B| \leq |S|+\sum_i |B_i| \leq |S| + \sum_i |R_i| + \sum _i |\Gamma(V_i)| \leq \frac {c_1 n}{r^\delta} + |R|+ \frac {c_1 n}{r^\delta}
= |R|+ \frac {2c_1 n}{r^\delta} = |R|+ \frac {|R|+|B|}{k^ \delta} 2c_1(c_2+c_3)^\delta = |R|+\varepsilon' (|R|+|B|)$.
We thus have,
\old{
\begin{eqnarray}
\lefteqn{|B| \leq |R| \cdot \frac {1+\varepsilon '}{1-\varepsilon '}} \nonumber \\
&=& |R| (1+\varepsilon ')(1+\varepsilon '+{\varepsilon '}^2 +...) \nonumber \\
&\leq& |R| (1+\varepsilon ')(1+\varepsilon ' + \varepsilon ') \nonumber \\
&=& |R| (1+3 \varepsilon ' + 2 {\varepsilon '}^2) \nonumber \\
&\leq& |R| (1+5 \varepsilon ') \nonumber \\
&=& |R| (1+\varepsilon) \ . \nonumber
\end{eqnarray}
}
\old{
\begin{align*}
|B| &\leq |R| \cdot \frac {1+\varepsilon '}{1-\varepsilon '} = |R| (1+\varepsilon ')(1+\varepsilon '+{\varepsilon '}^2 +...) \\
&\leq |R| (1+\varepsilon ')(1+\varepsilon ' + \varepsilon ') = |R| (1+3 \varepsilon ' + 2 {\varepsilon '}^2) \leq |R| (1+5 \varepsilon ') = |R| (1+\varepsilon) \ .
\end{align*}
}
$|B| \leq |R| \cdot \frac {1+\varepsilon '}{1-\varepsilon '} = |R| (1+\varepsilon ')(1+\varepsilon '+{\varepsilon '}^2 + \ldots)
\leq |R| (1+\varepsilon ')(1+\varepsilon ' + \varepsilon ') = |R| (1+3 \varepsilon ' + 2 {\varepsilon '}^2) \leq |R| (1+5 \varepsilon ') = |R| (1+\varepsilon)$.

\end{proof}

\subsection{Maximization problems}

Let us now recall the local search technique as it is used, e.g., in the context of geometric packing. Assume that we are considering a maximization problem ${\cal P}$, that is, one needs to find a maximum-cardinality subset of a given set $X$ that is a solution for ${\cal P}$. Let $X_0 \subseteq X$ be an initial not-necessarily-optimal solution for ${\cal P}$, and let $k$ be a sufficiently large constant. The local search technique checks whether there exists a subset $X' \subseteq X_0$ of size at most $k-1$ and a subset $X'' \subset X$ of size $k$, such that $(X_0 \setminus X') \cup X''$ is still a solution for ${\cal P}$. If yes, then it replaces $X'$ by $X''$ (i.e., it performs this local improvement) and resumes the search. Otherwise, it halts.

Let $B$ and $R$ be two solutions for ${\cal P}$, where $B$ was obtained by applying local search, and assume that $B \cap R = \emptyset$. (Otherwise, we can remove the elements that belong to both $B$ and $R$.)
Moreover, let $G=(V,E)$ be a graph, such that (i) $V = B \cup R$, and (ii) there exists an edge $e \in E$ between a vertex $b \in B$ and a vertex $r \in R$ if and only if $b$ and $r$ intersect. 

Theorem~\mbox{\ref{thm:ptas_max}} below states that if $G$ belongs to some family ${\cal F}$ that has a separator property, then $|B|\ge(1-\varepsilon)|R|$.
The proof is similar to the one in Section~\mbox{\ref{sect:minimization_problems}}; we include it for completeness.
We first construct an $r$-division of $G$, for $r = \frac{k}{c_2+c_3}$, and set $B_i= B \cap V_i$ and $R_i= R \cap V_i$.

\begin{myclaim}
\label{claim:substition_max}
For any index $i$, the set $(B \setminus (B_i\cup \Gamma(V_i))) \cup R_i$ is also a solution for~${\cal P}$.
\end{myclaim}
\begin{proof}
Fix $i$, and let $r \in R_i$. By definition, $r$ intersects only its neighbors in $G$, i.e., only elements in $(B_i \cup \Gamma(V_i)) \cup R_i$. Thus, for each $b \in B \setminus (B_i\cup \Gamma(V_i)))$, $r \cap b =\emptyset$.
Therefore, $(B \setminus (B_i\cup \Gamma(V_i))) \cup R_i$ is also a solution for ${\cal P}$.
\end{proof}

\begin{myclaim}
For any index $i$, $|R_i| \le |B_i|+|\Gamma(V_i)| < k$.
\end{myclaim}
\begin{proof}
By Claim~\ref{claim:substition_max}, the set $(B \setminus (B_i \cup \Gamma(V_i))) \cup R_i$ is also a solution for ${\cal P}$.
Moreover, by Claim~\ref{claim:size}, $|B_i \cup \Gamma(V_i)| = |B_i|+|\Gamma(V_i)| \leq |V_i| + |\Gamma (V_i)| < k$. So, if $|R_i|>|B_i \cup \Gamma(V_i)|$, then the local search algorithm would have replaced $B_i \cup \Gamma(V_i)$ by $R_i$ (or a subset of $R_i$ of size $k$, if $|R_i| > k$) before halting. Since it has not done so, we conclude that $|R_i| \le |B_i|+|\Gamma(V_i)|$.
\end{proof}

\begin{theorem}
\label{thm:ptas_max}
For any $\varepsilon$, $0 < \varepsilon < 1$, one can choose a constant $k$, such that
$|B|\ge(1-\varepsilon)|R|$.
\end{theorem}
\begin{proof}
Set $c'=2c_1(c_2+c_3)^\delta$, $k=(c'(\frac{2}{\varepsilon} -1))^{1/ \delta} $, and recall that $r=\frac {k}{c_2+c_3}$. Then,

\begin{align*}
|R| &\leq |S|+\sum_i |R_i| \leq |S| + \sum_i |B_i| + \sum _i |\Gamma(V_i)|
 \leq \frac {c_1 n}{r^\delta} + |B|+ \frac {c_1 n}{r^\delta} = |B|+ \frac {2c_1 n}{r^\delta} \\
&= |B|+ \frac {|R|+|B|}{k^ \delta} 2c_1(c_2+c_3)^\delta
= |B|+\frac{1}{\frac{2}{\varepsilon} -1} (|R|+|B|)
= |B|+\frac{\varepsilon}{2 -\varepsilon} (|R|+|B|) \ .
\end{align*}

Rearranging, we get that $|B|\ge(1-\varepsilon)|R|$.
\end{proof}


\section{Applications}
\label{sect:apps}

In this section we describe several original applications of our generalized local search technique.
The applications in Sections~\ref{sect:l_shallow} and~\ref{sect:triangle} use the maximization version and the applications in Section~\ref{sect:guarding} use the minimization version.

\subsection{Maximum $l$-shallow set for fat objects}
\label{sect:l_shallow}
We consider the \emph{maximum $l$-shallow set} problem for a set of fat objects. Let $D$ be a set of $n$ fat objects and let $l > 0$ be a constant. An \emph{$l$-shallow set} of $D$ is a subset $S$ of $D$, such that the depth of the arrangement of $S$ is at most $l$. In the maximum $l$-shallow set problem, one has to find a maximum-cardinality subset $S$ of $D$, such that $S$ is $l$-shallow. Notice that for $l=1$, $S$ is a maximum independent set of $D$.

\old{
Chan~\cite{Chan03} presented a PTAS for finding a maximum independent set of $D$. Later, Chan and Har-Peled~\cite{ChanH09} presented a local-search-based PTAS for finding a maximum independent set of a set of pseudo-disks (and also of $D$).
Notice that a maximum $l$-shallow set can be larger than the set that is obtained by repeatedly finding a maximum independent set of the set of remaining objects.
}
We show that our generalization of the local search technique enables us to apply local search to find a $(1-\varepsilon)$-approximation of a maximum $l$-shallow set of $D$.
Indeed, let $B$ be the set of fat objects obtained by applying local search and let $R$ be an optimal set. (We may assume that $B \cap R = \emptyset$, since otherwise, we can remove the objects that appear in both sets.) Notice that the analysis of Chan and Har-Peled~\cite{ChanH09} does not immediately apply here (assuming $l > 1$), since the intersection graph of $B \cup R$ is not necessarily planar (even if one removes the ``monochromatic'' edges)
However, this graph does have a separator of size $O(\sqrt{n})$.
This follows from results of Miller et al.~\cite{MillerTTV97} and Smith and Wormald~\cite{SmithW98}, who show that the intersection graph of a set of $l$-shallow fat objects has a separator of size $\sqrt{ln}$.

Recall that the local search algorithm begins with the empty solution, and in each iteration it improves the current solution by replacing a subset of size at most $k-1$ of the current solution with a larger subset of size at most $k$ of $D$, such that the resulting set is still a solution (i.e., it is still $l$-shallow).
Since $k$ is a constant, each iteration can be performed in polynomial time, and therefore the total running time of the local search algorithm is polynomial in $n$.

As for the size of $B$, since both $B$ and $R$ are $l$-shallow, the set $B \cup R$ is $2l$-shallow, and therefore the intersection graph of $B \cup R$ has a separator of size $\sqrt{2ln} = O(\sqrt{n})$.
Hence, by Theorem~\ref{thm:ptas_max}, we conclude that $|B| \ge (1-\varepsilon)|R|$.

Finally, we note that Chan's separator-based method~\cite{Chan03} can be used to obtain a $(1-\varepsilon)$-approximation of a maximum $l$-shallow set of $D$. However, as also mentioned in~\cite{ChanH09}, the disadvantage of his method is that it explicitly applies an algorithm for finding a separator, while we apply it only for analysis purposes.

\subsection{Maximum triangle matching in unit disk graphs}
\label{sect:triangle}
Next, we consider the \emph{maximum triangle matching} problem in unit disk graphs. Given a graph $G=(V,E)$, find a maximum-cardinality collection of pairwise-disjoint subsets $V_1,\ldots,V_m$ of $V$, each consisting of exactly 3 vertices, such that for each $V_i=\{u_i,v_i,w_i\}$, $1 \le i \le m$, the three edges $(u_i,v_i), (v_i,w_i), (w_i,u_i)$ belong to $E$.
Baker~\cite{B94} presented a PTAS for the important case where $G$ is planar.
We give a PTAS for the case where $G$ is an $l$-shallow unit disk graph. This PTAS also holds in more general settings; see remark below.

Let $S$ be a set of $n$ points in the plane, and let $\UDG(S)$ be the graph over $S$, in which there is an edge between $u$ and $v$ if and only if the Euclidean distance between $u$ and $v$ is at most 1.
Let ${\cal D}$ be the set of disks of radius $1/2$ centered at the points of $S$, and assume that the depth of the arrangement of ${\cal D}$ is at most some constant $l$. Then, there is an edge between $u$ and $v$ if and only if $d_u \cap d_v \ne \emptyset$, where $d_u, d_v$ are the disks of ${\cal D}$ centered at $u, v$, respectively.

Our generalization of the local search technique enables us to apply local search to find a $(1-\varepsilon)$-approximation of a maximum triangle matching of $\UDG(S)$.
We begin with the empty set of triangles, and in each iteration we replace a subset of at most $k-1$ triangles of the current solution with a larger subset of at most $k$ triangles, such that the resulting set is still a solution (i.e., it is still a triangle matching).

Consider the set of triangles ${\cal B}$ obtained by applying local search, and the set of triangles ${\cal R}$ of a maximum matching. (Each $\triangle \in {\cal B} \cup {\cal R}$ is a triangle $\triangle = \{a,b,c\}$,
where $a, b, c$ are points in $S$ and $(a,b), (b,c), (c,a)$ are edges of $\UDG(S)$.)
For each triangle $\triangle$ in ${\cal B}$ (resp. in ${\cal R}$), select an arbitrary unique representative point $p_{\triangle}$ that lies in the interior of $\triangle$, and denote by $B$ (resp. $R$) the resulting set of representative points.

Now, consider the graph $G=(B \cup R,E)$,    
where $(p,q) \in E$ if and only if the Euclidean distance between $p$ and $q$ is at most 2.
Notice that $G$ satisfies the following locality condition. For any two points $b \in B$ and $r \in R$, whose corresponding triangles are $\triangle u_bv_bw_b$ and $\triangle u_rv_rw_r$, respectively, and such that $\{u_b, v_b, w_b\} \cap \{u_r,  v_r,  w_r\} \neq \emptyset$, the edge $(b,r)$ belongs to $E$.
This is true since the distance from $b$ (alt., $r$) to any of the its triangle corners is at most 1.


It remains to show that $G$ has a separator.
Let ${\cal D'}$ be the set of unit disks centered at the points of $B \cup R$. Then, $(p,q) \in E$ if and only if the disks of ${\cal D'}$ centered at $p$ and $q$, respectively, intersect. It is easy to see that the depth of the arrangement of ${\cal D'}$ is at most some constant $c = c(l)$. This follows from the assumption that the depth of the arrangement of ${\cal D}$ is at most $l$.
Therefore, by Miller et al.~\cite{MillerTTV97}, $G$ has a separator of size $O(\sqrt{cn}) = O(\sqrt{n})$.
Hence, by Theorem~\ref{thm:ptas_max}, we conclude that $|B| \ge (1-\varepsilon)|R|$.

\old{
\begin{myclaim}
The depth of the arrangement of ${\cal D'}$ is at most some constant $c = c(l)$.
\end{myclaim}
\begin{proof}
Follows from the assumption that the depth of the arrangement of ${\cal D}$ is at most $l$.
\old{
Every point in the plane is covered by at most $l$ disks of ${\cal D}$.
For each $s \in S$, the number of points in $S \setminus \{s\}$ at distance $1$ or less from $s$ (i.e., the degree of $s$ in $\UDG(S)$), is at most $l-1$.
Therefore, each point in the plane is of distance at most $2$ from $k(k-1)$ points in $S$. In particular, each point in the plane is of distance at most $2$ from $k(k-1)$ vertices in $V$.
Thus, each point in the plane is covered by at most $k(k-1)$ disks in ${\cal D}_2$.
}
\end{proof}
}
\old{
\begin{corollary}
$G$ has a separator of size $\sqrt{cn} = O(\sqrt{n})$.
\end{corollary}
}
\old{
The claim implies that $G$ has a separator of size $O(\sqrt{cn}) = O(\sqrt{n})$.
Hence, by Theorem~\ref{thm:ptas_max}, we conclude that $|B| \ge (1-\varepsilon)|R|$.
}
\noindent
{\bf Remark}.
It is easy to see that one can replace ``triangle'' in the above result with any connected graph (i.e., structure) $H$ with a constant number of vertices, and obtain a PTAS for maximum $H$-matching in unit disk graphs.

\old{
\rom{The following section is not well written -- I spent a lot of time trying to formulate it.... But still I don't feel that it is good enough.}

By a similar construction one can also obtain a PTAS for $G'$ matching of $\UDG(S)$, where $G'$ is a graph of constant diameter $d$, and constant number of nodes.
E.g., partition the graph into maximum number of cycles of length $7$ (i.e., $C_7$, $d=6$). In each iteration of the algorithm we try to improve the solution by removing at most $k-1$ nodes and adding $k+6$ nodes of $S$. The number of subsets considered remains polynomial and the verification at each step also remains polynomial. Thus the total running time remains polynomial.
Now, each $x \in \cal{B} \cup \cal{R}$ is isomorphic to $C_7$. Each vertex in $B \cup R$ correspond to a different cycle, and the vertex is located inside its convex-hull.
Thus, for each $b \in B$ (alt., $r \in R$), the Euclidean distance to any point on the cycle is at most $d$. Therefore, we consider the graph $G=(R \cup B,E)$, in which $(u,v) \in E$ if and only if the Euclidean distance between $u$ and $v$ is at most $2d$.
This graph has the property that for every two cycles of length $7$ in $\cal{B} \cup \cal{R}$ that intersect, there is an edge in $G$.  Since $k$ and $d$ are constants, then
the depth of the arrangement of disks centered at $R \cup B$ and with distance $2d$ is at most $k^{2d}$. Thus, there exists an $O(\sqrt{n})$ separator for $G$, and by Theorem~\ref{thm:ptas_max}, the local search algorithm yields a PTAS for this matching problem.
}

\subsection{Guarding with limited visibility}
\label{sect:guarding}
We now demonstrate our generalization of the local search technique on a family of \emph{covering} problems, or, more precisely, on a family of \emph{guarding} problems. Throughout this section, $\cal{G}$ denotes a set of points, representing stationary guards, where each guard $g \in \cal{G}$ has its own range of sight $r_g$.
Let $D_g$ denote the disk of radius $r_g$ centered at $g$, and set $D = \{ D_g : g \in \cal{G} \}$.
We assume that the depth of the arrangement of $D$ is at most some constant $l$.

\paragraph{Guarding a polygon.}
Given a polygon $P$ (possibly with holes) and a set ${\cal G} \subseteq P$,
a \emph{minimum guarding set} for $P$ (with respect to $\cal{G}$) is a minimum-cardinality subset $\cal{G}'$ of $\cal{G}$, such that every $p \in P$ is guarded by $\cal{G}'$, i.e., for every $p \in P$, there exists $g \in \cal{G}'$, such that $|gp| \le r_g$ and the segment $\overline{gp}$ is contained in $P$.

Eidenbenz et al.~\cite{ESW01} proved that finding a minimum guarding set for a polygon, where the given set of guards is the polygon's set of vertices, is APX-hard, even if the polygon has no holes and there are no limits on the ranges.
We show that our generalization of the local search technique enables us to apply local search to find a $(1+\varepsilon)$-approximation of a minimum guarding set for $P$ (with respect to $\cal{G}$), under the above assumption concerning the depth of the arrangement of $D$.

For a guard $g \in \cal{G}$, let $\VP{g}$ denote its visibility polygon and let $\VR{g}$ denote its visible region, where $\VR{g} = \VP{g} \cap D_g$. Notice that $\VR{g}$ is not necessarily convex.
For a subset $X \subseteq \cal{G}$, let $D_X$ denote the set $\{D_g : g \in X\}$ and
let $\VR{X}$ denote the set $\{\VR{g} : g \in X\}$.
Let $B \subseteq \cal{G}$ be the set of guards obtained by applying local search (to the set $\VR{\cal{G}}$) and let $R$ be an optimal set. (We may assume that $B \cap R = \emptyset$, since otherwise, we can remove the guards that appear in both sets.) Consider the
graph $G=(B \cup R, E)$, where $(g,g') \in E$ if and only if $\VR{g} \cap \VR{g'} \ne \emptyset$.
We claim that $G$ has a separator of size $\sqrt{ln} = O(\sqrt{n})$, since any separator of the intersection graph of $D_B \cup D_R$ is also a separator of $G$, and the former graph has a separator of size $\sqrt{ln} = O(\sqrt{n})$, by a result of Miller et al.~\cite{MillerTTV97}.
Observe that the locality condition is satisfied, since for any $p \in P$, let $b \in B$ and $r \in R$ be two guards that guard $p$, then $\VR{b} \cap \VR{r} \ne \emptyset$ and hence the edge $(b,r)$ is in $G$. Therefore, by Theorem~\ref{thm:ptas}, we conclude that $|B| \le (1+\varepsilon)|R|$.

\begin{theorem}
There exists a PTAS for minimum guarding a (not necessarily simple) polygon with respect to a given set of guards, assuming that the depth of the corresponding arrangement of disks is bounded by a constant.
\end{theorem}


It is possible to consider other versions of the problem in which the visible region of a guard $g \in \cal{G}$ is defined differently, as long as $\VR{g} \subseteq D_g$. The approximation analysis for these versions is the same as for the standard version.
For example, consider the problem of guarding a polygon $P$, where each guard can see through at most some fixed number of walls (edges of the polygon $P$). Then, the visibility polygon of a guard $g \in \cal{G}$ is different than its visibility polygon in the standard version, but still its visible region is contained in $D_g$. For such problems, one needs to modify the verification procedure of the local search algorithm (i.e., whether a given set of guards consists of a solution). We thus obtain, for example, the following corollary.

\begin{corollary}
There exists a PTAS for minimum guarding a polygon through walls, assuming that the depth of the corresponding arrangement of disks is bounded by a constant.
\end{corollary}

\old{
\paragraph{Guarding a polygon through walls.}
The following problem is motivated by the problem of placing wireless transceivers inside of a building, where a client can receive signals even from transceivers that are located in other rooms as long as the number of walls between them is small.
That is, we consider the problem of guarding a polygon $P$ such that the guards can see through at most some number $m$ of walls.
More formally, given a polygon $P$ and a set of points $\cal{G}$ in $P$, where each guard $g \in \cal{G}$  has a limited guarding range $r_g$. Find a minimum cardinality subset $\cal{G}' \subseteq \cal{G}$ such that by placing guards in $\cal{G}'$, we have that for every $p \in P$ there exists $g \in \cal{G}'$ such that segment $|pg| \leq r_g$ and $|pg|$ intersect with at most $m$ edges of $P$.

Again, if we assume that $d$ is a constant, then we are able to apply the local search algorithm.
In each iteration, the algorithm checks if a solution is valid by checking that the guards of the solution guard $P$ (under the constraint that a guard can see through at most $m$ walls).
Let $\VR{g}$ denote the visible region of guard $g \in \cal{G}$. For a subset $X \subseteq \cal{G}$, let $\VR{X}$ denote the set $\{\VR{g} : g \in X\}$.
Consider the intersection graph of $\VR{B} \cup \VR{R}$, which has a separator of $O(\sqrt{n})$ by similar claims as above.
Thus, the local search is a $(1+\varepsilon)$-approximation.


\begin{corollary}
There exists a PTAS for the minimum guarding set of a polygon through a fixed number of walls, assuming that the ranges are limited and the depth of the arrangement is at most $O(1)$.
\end{corollary}
}

\paragraph{Terrain guarding with limited visibility.}
Let $T$ be a 1.5D terrain (i.e., an $x$-monotone polygonal chain), let ${\cal G} \subseteq T$ be a finite set of guards, and let $X \subseteq T$ be a finite set of points to be guarded. A guard $g \in \cal{G}$ sees a point $x \in X$ if $\overline{gx}$ does not cross any edge of $T$ and $|gx| \le r_g$. The goal is to find a minimum-cardinality subset $\cal{G}' \subseteq \cal{G}$, such that $\cal{G}'$ sees $X$ (i.e., for each point $x \in X$, there exists a guard in $\cal{G}'$ that sees $x$).
Assuming unlimited visibility (i.e., $r_g= \infty$, for each $g \in \cal{G}$), the problem is known to be NP-hard~\cite{CEU95,KK10}, and there exists a local-search-based PTAS for it~\cite{GKKV09}.  
This PTAS relies heavily on the, so called, ``order claim'', which states that for any four points $a, b, c, d$ on $T$ (listed from left to right), if $a$ sees $c$ and $b$ sees $d$ then $a$ also sees $d$. Unfortunately, this claim is false in the limited visibility version.
However, this version is a simple variant of polygon guarding with limited visibility.
Therefore, we have
\begin{theorem}
There exists a PTAS for terrain guarding, assuming that the depth of the corresponding arrangement of disks is bounded by a constant.
\end{theorem}

Consider the guarding problems above, but now set $r_g=1$, for each $g \in \cal{G}$, and assume the requirement is to guard a finite set $X$ of points in $P$ (alt., on $T$). Then,
one may replace the assumption that the depth of the arrangement of disks is bounded by a constant with the assumption that the points in $X$ are not too crowded, i.e., the distance between any two points in $X$ is at least some constant $\delta > 0$. This assumption is sometimes more convenient. Under this assumption, each guard sees at most some constant number $c=c(\delta)$ of points of $X$, and it is likely therefore that there are many redundant guards. (A guard is \emph{redundant} if there exists another guard that sees the same subset of points of $X$.) We thus remove redundant guards from $\cal{G}$, one at a time, as long as there are such guards in $\cal{G}$. It is easy to see that the depth of the arrangement of disks corresponding to the set of remaining guards is bounded by a constant. We obtain, e.g., the following theorem.

\begin{theorem}
There exists a PTAS for minimum guarding a sparse set of points $X$ within a (not necessarily simple) polygon $P$ with respect to a given set of guards of unit visibility range.
\end{theorem}

\section{Discrete coverage of points}\label{sec:coverage}

We consider the following fundamental problem.
Let $P$ be a set of points in the plane and let $D$ be a set of disks that covers $P$. (We assume that the centers of the disks in $D$ are in general position.) Find a minimum-cardinality subset $D' \subseteq D$ that covers $P$.

If $D$ is a set of unit disks, then this is the dual version of the problem known as discrete piercing of unit disks. Mustafa and Ray~\cite{MR10} presented a PTAS for the latter problem (even for the case of varying radii).
We present a local-search-based PTAS for discrete coverage of points. Our proof is based on the framework developed in~\cite{MR10} and in Section \ref{sec:general}, and is inspired by the work of Gibson and Pirwani~\cite{GP10}.

\begin{theorem}
\label{thm:diskcoverage}
There exists a PTAS for discrete coverage of points by disks (alt., by axis-parallel squares).
\end{theorem}
\begin{proof}
We may assume that there do not exist two disks in $D$, such that one of them is contained in the other, since, if two such disks do exist, we can always prefer the larger one.
Let $B$ be the set of the centers of the disks returned by the local search algorithm, and let $R$ be the set of the centers of the disks in an optimal solution.
We may assume that $B \cap R = \emptyset$ (since otherwise, we can remove the centers that belong to both sets).
For a center $c \in B \cup R$, let $D_c \in D$ denote the corresponding disk and $r_c$ its radius.
In order to prove that $|B| \leq (1+\varepsilon)|R|$, we need to present a planar bipartite graph $G=(B \cup R,E)$ satisfying the following locality condition: For each $p \in P$ there exist $b \in B$ and $r \in R$, such that $ p \in D_b$, $p \in D_r$, and $(b,r) \in E$.

Consider the additively weighted Voronoi diagram of $B \cup R$, constructed according to the distance function $\delta(p, c) = d(p, c) - r_c$.
Observe first that for any center $c \in B \cup R$ and point $p \in \reals^2$, if $\delta(p, c) \leq 0$, then $p \in D_c$ (and vice versa).
Now, let $cell(c)$ denote the cell of the diagram corresponding to $c \in B \cup R$.
It is well known that for each $c \in B \cup R$, $c \in cell(c)$ and $cell(c)$ is connected.

\old{
We prove two simple and known claims concerning the diagram.
\gila{I think we can say here ``we claim that ... . Indeed, ..... .'' instead of using enumerated claims, moreover, if we say that the claims are known, we do not need to prove them, right?!.}
\begin{myclaim}\label{lemma:cell}
For each $c \in B \cup R$, $c \in cell(c)$.
\end{myclaim}
\begin{proof}
Assume that $c \notin cell(c)$, for some $c \in B \cup R$, and let $c' \in B \cup R$, such that $c \in cell(c')$. Then,
$\delta(c, c') < \delta(c, c)$, or $d(c,c')-r_{c'} < -r_c$, which implies that $D_c \subseteq D_{c'}$ --- a contradiction.
\end{proof}

\begin{myclaim}\label{lemma:connected}
For each $c \in B \cup R$, $cell(c)$ is connected.
\end{myclaim}
\begin{proof}
Assume that $cell(c)$ is not connected, for some $c \in B \cup R$. Then, $cell(c)$ consists of two or more connected regions, such that any two of them are disconnected. Let $C_1$ be the region of $cell(c)$, such that $c \in C_1$. Let $C_2$ be any other region of $cell(c)$ and let $q \in C_2$. The line segment $\segment{q}{c}$ must pass through the interior of some other cell $cell(c')$. Let $p \in cell(c')$ be a point on $\segment{q}{c}$ and in the interior of $cell(c')$.
Then, $\delta(p,c') < \delta(p,c)$ and $\delta(q,c) < \delta(q,c')$.
Therefore, $d(p,c') - r_{c'} < d(p,c) - r_c = d(q,c) - d(q,p) - r_c = \delta(q,c) - d(q,p) < \delta(q,c') - d(q,p) = d(q,c') - r_{c'} - d(q,p) \leq d(p,c') - r_{c'}$, which is of course impossible.
\end{proof}
}

Our planar bipartite graph is the dual graph of the weighted Voronoi diagram defined above, without the monochromatic edges.
That is, there is an edge between two centers $c$ and $c'$ in $B \cup R$, if and only if $c \in B$ and $c' \in R$, or vice versa, and their cells are adjacent to each other. We denote this graph by $G=(R \cup B, E)$. It is easy to see that $G$ is planar. Indeed, by the diagram's properties, any edge $(c,c') \in E$ can be drawn such that it is contained in $cell(c) \cup cell(c')$, and within each cell $cell(c)$, it is easy to ensure that the edges do not cross each other.

We now show that $G$ satisfies the locality condition.
Let $p \in P$, and assume w.l.o.g. that $p \in cell(r)$, where $r \in R$. Let $b \in B$ be the closest center to $p$ according to $\delta$. That is, for any $b' \in B$, $\delta(p, b) \leq \delta(p, b')$.
Notice that $\delta(p, b) \le 0$, since there exists a center $b' \in B$ whose disk covers $p$, and therefore, $\delta(p, b') \leq 0$. We conclude by the observation
above that $p \in D_b$.
Consider the cells we visit when walking along the line segment $\segment{p}{b}$ from $p$ to $b$.
Since $p \notin cell(b)$ and $b \in cell(b)$, we must at some point enter $cell(b)$.
Let $c \in B \cup R$ be the center for which $cell(c)$ is the last cell that we visit before entering $cell(b)$, and let $q$ be the point on $\segment{p}{b}$, which is also on the boundaries of $cell(c)$ and $cell(b)$.

It remains to show that $c \in R$, implying that $(b,c) \in E$, and that $p \in D_c$.
Using the triangle inequality and since the centers of the disks are in general position, we get that $\delta(p,c) = d(p,c) - r_c < d(p,q) + d(q,c) - r_c =  d(p,q) + \delta(q,c)$. But, $d(p,q) + \delta(q,c) = d(p,q) + \delta(q,b) = d(p,q) + d(q,b) - r_b = d(p,b) - r_b = \delta(p,b)$, so we get that
$\delta(p,c) < \delta(p,b)$. Now, since $b$ is the closest center to $p$ among the centers in $B$, we conclude that $c \in R$, and since $\delta(p,b) \le 0$, we conclude that $p \in D_c$.
The axis-parallel version is obtained by replacing the $L_2$ metric by the $L_{\infty}$ metric.
\end{proof}

\subsection{Discrete coverage of a polygon}
Consider now the following problem. Let $Q$ be a polygon, and let $D$ be a set of disks (alt., a set of axis-parallel rectangles). Find a minimum-cardinality subset of $D$ that covers $Q$. The local search algorithm can be easily adapted to this setting. At each iteration of the algorithm, instead of checking whether all points are covered, one needs to check whether the entire polygon is covered. This can be done in polynomial time.
The analysis is essentially the same as for discrete coverage of points, except that here $p$ is any point in the polygon $Q$.
We thus conclude that
\begin{theorem}
There exists a PTAS for discrete coverage of a polygon by disks (alt., by axis-parallel squares).
\end{theorem}

\subsection{The class cover problem}

The class cover problem is defined as follows: Let $B$ be a set of blue points and let $R$ be a set of red points and set $n = |B|+|R|$. Find a minimum-cardinality set $D$ of disks (alt., of axis-parallel squares) that covers the blue points, but does not cover any of the red points. That is, find a minimum-cardinality set $D$, such that $B \cap \cup_{d \in D} d = B$ and $R \cap \cup_{d \in D} d = \emptyset$.
Bereg et al.~\cite{BeregCDPSV10} study several versions of the class cover problem with boxes. In particular, they prove that the class cover problem  with axis-parallel squares is NP-hard and give an $O(1)$-approximation algorithm for this version. We show that the class cover problem with squares (resp., disks) is essentially equivalent to discrete coverage of points by squares (resp., disks), and therefore there exists a PTAS for both versions.

Indeed, consider any set $D$ of disks that cover the points in $B$ and does not cover any point in $R$. It is easy to see that one can replace each disk $d \in D$ with a ``legal'' disk $d'$, such that $B \cap d \subseteq B \cap d'$ and either there are three points on $d'$'s boundary, or there are exactly two points on $d'$'s boundary and the line segment between them is a diameter of $d'$. Therefore, we can transform the class cover problem with disks to discrete coverage of points by disks, by first computing all $O(n^3)$ disks defined by either triplets or pairs of points and then removing those that are ``illegal''. Similarly, we can transform the class cover problem with squares to discrete coverage of points by squares. We thus obtain
\begin{theorem}
There exists a PTAS for the class cover problem with disks (alt., with axis-parallel squares).
\end{theorem}

\noindent
{\bf Concluding remark.}
Very recently it has been brought to our attention that Chan and Grant~\cite{ChanGrant} observe that the PTAS of Mustafa and Ray~\cite{MR10} for discrete hitting set of half-spaces in $\reals^3$ implies a PTAS for discrete coverage of points by disks, by a lifting transformation that maps disks to lower half-spaces in $\reals^3$ and by duality between points and half-spaces.
However, our PTAS above for discrete coverage of points by disks is direct and refrains from moving to $\reals^3$.
 
\old{
Bereg et al.~\cite{BeregCDPSV10} claim that by a result of Cannon and Cowen~\cite{CannonC04} there is a PTAS for the class cover problem with disks. However, Cannon and  Cowen~\cite{CannonC04} only consider the case where the disks are with equal radius, i.e., unit disks, and it is not clear how their result yields a PTAS for the general case. For the class coverage problem with squares Bereg et al.~\cite{BeregCDPSV10} show an $O(1)$-approximation. We show that since the class coverage problem is equivalent to discrete coverage problem, thus the local search algorithm is a PTAS for the class coverage with either disks or axis-parallel squares.
Notice that for any solution $D'$ for the class coverage problem with disks, one can expand each disk $d \in D'$, without changing the points that it covers and without adding red points, such that: (i) $d$  contains three red points on its boundary. (ii) $d$ contains two red points on its boundary and no additional expansion will add any blue point, i.e. such a disks can be defined by two red points and one blue.
Therefore, only $O(n^3)$ disks should be considered when solving the class coverage problem. In order to solve the problem, one should compute all $n \choose 3$ disks defined by three red points and $ {n \choose 2} \cdot {n \choose 1}$ disks defined by two red points and some blue point, and choose the maximal by containment of blue points disks which doesn't contain a red point. Then, one should compute a minimum cardinality subset of those disks  that covers $B$. This is the discrete coverage problem, for which in Section~\ref{sec:coverage} we show a PTAS.
}


\old{
\newpage
\appendix
\section{Proof of Theorem 5}\label{app:max}
We prove below that if $G$ belongs to some family ${\cal F}$ that has a separator property, then $|B|$ is not much smaller than $|R|$.
We first construct an $r$-division of $G$, for $r = \frac{k}{c_2+c_3}$, and set $B_i= B \cap V_i$ and $R_i= R \cap V_i$.

\begin{myclaim}
\label{claim:substition_max}
For any index $i$, the set $(B \setminus (B_i\cup \Gamma(V_i))) \cup R_i$ is also a solution for ${\cal P}$.
\end{myclaim}
\begin{proof}
Fix $i$, and let $r \in R_i$. By definition, $r$ intersects only its neighbors in $G$, i.e., only elements in $(B_i \cup \Gamma(V_i)) \cup R_i$. Thus, for each $b \in B \setminus (B_i\cup \Gamma(V_i)))$, $r \cap b =\emptyset$.
Therefore, $(B \setminus (B_i\cup \Gamma(V_i))) \cup R_i$ is also a solution for ${\cal P}$.
\end{proof}

\old{
\begin{myclaim}
\label{claim:size_max}
For any index $i$, $|V_i| + |\Gamma (V_i)| < k$.
\end{myclaim}
\begin{proof}
\begin{eqnarray}
\lefteqn{|V_i \cup \Gamma(V_i) |\leq c_2 r + c_3 r^{1-\delta} \leq \frac{c_2 k}{c_2 +c_3} + \frac{c_3 k^{1-\delta}} {(c_2+c_3)^{1-\delta}}} \nonumber \\
&=& \frac{c_2 k + c_3  (c_2 +c_3)^\delta  k^{1-\delta}}{c_2+c_3} < \frac {c_2 k + c_3 k}{c_2 + c_3}  = k \ . \nonumber
\end{eqnarray}
\end{proof}
}

\begin{myclaim}
For any index $i$, $|R_i| \le |B_i|+|\Gamma(V_i)| < k$.
\end{myclaim}
\begin{proof}
By Claim~\ref{claim:substition_max}, the set $(B \setminus (B_i \cup \Gamma(V_i))) \cup R_i$ is also a solution for ${\cal P}$.
Moreover, by Claim~\ref{claim:size}, $|B_i \cup \Gamma(V_i)| = |B_i|+|\Gamma(V_i)| \leq |V_i| + |\Gamma (V_i)| < k$. So, if $|R_i|>|B_i \cup \Gamma(V_i)|$, then the local search algorithm would have replaced $B_i \cup \Gamma(V_i)$ by $R_i$ (or a subset of $R_i$ of size $k$, if $|R_i| > k$) before halting. Since it has not done so, we conclude that $|R_i| \le |B_i|+|\Gamma(V_i)|$.
\end{proof}

\noindent {\bf Theorem~\ref{thm:ptas_max}.} 
{\it
For any $\varepsilon$, $0 < \varepsilon < 1$, one can choose a constant $k$, such that
$|B|\ge(1-\varepsilon)|R|$.
}

\begin{proof}
Set $c'=2c_1(c_2+c_3)^\delta$, $k=(c'(\frac{2}{\varepsilon} -1))^{1/ \delta} $, and recall that $r=\frac {k}{c_2+c_3}$. Then,
\old{
\begin{eqnarray}
\lefteqn{|R|\leq |S|+\sum_i |R_i| \leq |S| + \sum_i |B_i| + \sum _i |\Gamma(V_i)|} \nonumber \\
& \leq& \frac {c_1 n}{r^\delta} + |B|+ \frac {c_1 n}{r^\delta} = |B|+ \frac {2c_1 n}{r^\delta} \nonumber \\
&=& |B|+ \frac {|R|+|B|}{k^ \delta} 2c_1(c_2+c_3)^\delta  \nonumber \\
&=& |B|+\frac{1}{\frac{2}{\varepsilon} -1} (|R|+|B|) \ \nonumber \\
&=& |B|+\frac{\varepsilon}{2 -\varepsilon} (|R|+|B|) \ . \nonumber
\end{eqnarray}
}
\begin{align*}
|R| &\leq |S|+\sum_i |R_i| \leq |S| + \sum_i |B_i| + \sum _i |\Gamma(V_i)|
 \leq \frac {c_1 n}{r^\delta} + |B|+ \frac {c_1 n}{r^\delta} = |B|+ \frac {2c_1 n}{r^\delta} \\
&= |B|+ \frac {|R|+|B|}{k^ \delta} 2c_1(c_2+c_3)^\delta
= |B|+\frac{1}{\frac{2}{\varepsilon} -1} (|R|+|B|)
= |B|+\frac{\varepsilon}{2 -\varepsilon} (|R|+|B|) \ .
\end{align*}

Rearranging, we get that $|B|\ge(1-\varepsilon)|R|$.
\end{proof}

\section{Proof of Theorem 9}\label{app:cover}

We may assume that there do not exist two disks in $D$, such that one of them is contained in the other, since, if two such disks do exist, we can always prefer the larger one.
Let $B$ be the set of the centers of the disks returned by the local search algorithm, and let $R$ be the set of the centers of the disks in an optimal solution.
We may assume that $B \cap R = \emptyset$ (since otherwise, we can remove the centers that belong to both sets).
For a center $c \in B \cup R$, let $D_c \in D$ denote the corresponding disk and $r_c$ its radius.
In order to prove that $|B| \leq (1+\varepsilon)|R|$, we need to present a planar bipartite graph $G=(B \cup R,E)$ satisfying the following locality condition: For each $p \in P$ there exist $b \in B$ and $r \in R$, such that $ p \in D_b$, $p \in D_r$, and $(b,r) \in E$.

Consider the additively weighted Voronoi diagram of $B \cup R$, constructed according to the distance function $\delta(p, c) = d(p, c) - r_c$.
Observe first that for any center $c \in B \cup R$ and point $p \in \reals^2$, if $\delta(p, c) \leq 0$, then $p \in D_c$ (and vice versa).
Now, let $cell(c)$ denote the cell of the diagram corresponding to $c \in B \cup R$. We prove two simple and known claims concerning the diagram.

\begin{myclaim}\label{lemma:cell}
For each $c \in B \cup R$, $c \in cell(c)$.
\end{myclaim}
\begin{proof}
Assume that $c \notin cell(c)$, for some $c \in B \cup R$, and let $c' \in B \cup R$, such that $c \in cell(c')$. Then,
$\delta(c, c') < \delta(c, c)$, or $d(c,c')-r_{c'} < -r_c$, which implies that $D_c \subseteq D_{c'}$ --- a contradiction.
\end{proof}

\begin{myclaim}\label{lemma:connected}
For each $c \in B \cup R$, $cell(c)$ is connected.
\end{myclaim}
\begin{proof}
Assume that $cell(c)$ is not connected, for some $c \in B \cup R$. Then, $cell(c)$ consists of two or more connected regions, such that any two of them are disconnected. Let $C_1$ be the region of $cell(c)$, such that $c \in C_1$. Let $C_2$ be any other region of $cell(c)$ and let $q \in C_2$. The line segment $\segment{q}{c}$ must pass through the interior of some other cell $cell(c')$. Let $p \in cell(c')$ be a point on $\segment{q}{c}$ and in the interior of $cell(c')$.
Then, $\delta(p,c') < \delta(p,c)$ and $\delta(q,c) < \delta(q,c')$.
Therefore, $d(p,c') - r_{c'} < d(p,c) - r_c = d(q,c) - d(q,p) - r_c = \delta(q,c) - d(q,p) < \delta(q,c') - d(q,p) = d(q,c') - r_{c'} - d(q,p) \leq d(p,c') - r_{c'}$, which is of course impossible.
\end{proof}

\paragraph{A planar bipartite graph.}
Our graph is the dual graph of the weighted Voronoi diagram defined above, without the monochromatic edges.
That is, there is an edge between two centers $c$ and $c'$ in $B \cup R$, if and only if $c \in B$ and $c' \in R$, or vice versa, and their cells are adjacent to each other. We denote this graph by $G=(R \cup B, E)$. It is easy to see that $G$ is planar. Indeed, by Claim~\ref{lemma:connected}, any edge $(c,c') \in E$ can be drawn such that it is contained in $cell(c) \cup cell(c')$, and within each cell $cell(c)$, it is easy to ensure that the edges do not cross each other.

We now show that $G$ satisfies the locality condition.
Let $p \in P$, and assume w.l.o.g. that $p \in cell(r)$, where $r \in R$. Let $b \in B$ be the closest center to $p$ according to $\delta$. That is, for any $b' \in B$, $\delta(p, b) \leq \delta(p, b')$.
Notice that $\delta(p, b) \le 0$, since there exists a center $b' \in B$ whose disk covers $p$, and therefore, $\delta(p, b') \leq 0$. We conclude, by the observation just above Claim~\ref{lemma:cell}, that $p \in D_b$.
Consider the cells we visit when walking along the line segment $\segment{p}{b}$ from $p$ to $b$.
Since $p \notin cell(b)$ and $b \in cell(b)$ (Claim~\ref{lemma:cell}), we must at some point enter $cell(b)$.
Let $c \in B \cup R$ be the center for which $cell(c)$ is the last cell that we visit before entering $cell(b)$, and let $q$ be the point on $\segment{p}{b}$, which is also on the boundaries of $cell(c)$ and $cell(b)$.

It remains to show that $c \in R$, implying that $(b,c) \in E$, and that $p \in D_c$.
Using the triangle inequality and since the centers of the disks are in general position, we get that $\delta(p,c) = d(p,c) - r_c < d(p,q) + d(q,c) - r_c =  d(p,q) + \delta(q,c)$. But, $d(p,q) + \delta(q,c) = d(p,q) + \delta(q,b) = d(p,q) + d(q,b) - r_b = d(p,b) - r_b = \delta(p,b)$, so we get that
$\delta(p,c) < \delta(p,b)$. Now, since $b$ is the closest center to $p$ among the centers in $B$, we conclude that $c \in R$, and since $\delta(p,b) \le 0$, we conclude that $p \in D_c$.
The following theorem summarizes our result, where the axis-parallel version is obtained by replacing the $L_2$ metric by the $L_{\infty}$ metric.

\vspace{5mm}\noindent {\bf Theorem~\ref{thm:diskcoverage}.} 
{\it 
There exists a PTAS for discrete coverage of points by disks (alt., by axis-parallel squares).
}
}

\old{
\begin{theorem}\label{thm:diskcoverage_appendix}
There exists a PTAS for discrete coverage of points by disks.
\end{theorem}
Finally, by using the $L_{\infty}$ metric instead of the $L_2$ metric in the proof of Theorem~\ref{thm:diskcoverage}, we also obtain that
\begin{theorem}
There exists a PTAS for discrete coverage of points by axis-parallel squares.
\end{theorem}
}

\end{document}